\documentclass[conference]{IEEEtran}
\IEEEoverridecommandlockouts
\usepackage{comment}
\usepackage{caption}
\newtheorem{theorem}{Theorem}
\usepackage{subcaption}

\newtheorem{lemma}{\textbf{Lemma}}

\newtheorem{corollary}{\textbf{Corollary}}

\usepackage[left=1.6cm,right=1.6cm,top=1.9cm]{geometry}
\usepackage{cite}
\usepackage{amsmath,amssymb,amsfonts}
\usepackage{algorithmic}
\usepackage{graphicx}
\usepackage{textcomp}
\usepackage{xcolor}
\def\BibTeX{{\rm B\kern-.05em{\sc i\kern-.025em b}\kern-.08em
    T\kern-.1667em\lower.7ex\hbox{E}\kern-.125emX}}

\begin{document}
%

\title{Coverage Analysis in Terahertz Clustered HetNets
}

\author{Hadeel Obaid\\\
\IEEEauthorblockA{\textit{Faculty of Information Technology and Communication Sciences}, Tampere University, Finland \\ 
hadeel.obaid@tuni.fi}
}
\maketitle


\begin{abstract}
Terahertz (THz) transmission technologies hold significant potential for enabling ultra-broadband,  short-range communication in next-generation networks.  Despite the vast bandwidth, 
THz signals suffer from limited transmission range, and a feasible scenario is to deploy
THz within clustered heterogeneous networks (HetNets) to enhance coverage.
This paper investigates THz communication in clustered HetNets, leveraging stochastic geometry for performance analysis.  Specifically, we consider two tiers of macro base stations (MBS) and small base stations (SBS). The MBS tier is modeled as a Poisson Point Process (PPP), and both the SBS tier and users are modeled as a Poisson Cluster Process (PCP) to capture user clustering and network hotspots. We derive the analytical expressions for user association probabilities, the Laplace transform of interference, and the coverage probability.  The derived coverage probability is validated through Monte Carlo simulation. The numerical results show that the coverage in THz PCP-HetNets is higher than that achieved in THz PPP HetNets. In addition, a moderate spatial spread of SBSs is beneficial for coverage.  
\end{abstract}


\maketitle
\vspace{-0.2cm}

\section{Introduction}
\vspace{-0.1cm}


\IEEEPARstart{T}he envisioned 6G system aims to deliver high-speed data services with minimal latency, over a higher frequency spectrum. The Terahertz (THz) band, ranging from \textcolor{black}{0.1 to 10 THz}, offers an extensive bandwidth, making it a promising candidate for achieving unprecedented data rates in next-generation networks.  However, despite its wide bandwidth, the THz signal's short transmission range presents a major challenge for deploying THz communication systems \cite{obaid2024collaborative}. The limited link distances in THz communications hinder coverage and result in significant signal attenuation over extended distances. The propagation of THz signals faces three key challenges: molecular absorption, high free-space attenuation, and low transmission power \textcolor{black}{\cite{obaid2024collaborative}}. These limitations make it difficult to achieve universal coverage and reliable connectivity with standalone THz base stations. 
 Furthermore, the short range of THz signals makes conventional wide-area deployments inefficient, leading to coverage gaps and reduced network reliability in urban and hotspot environments.

A feasible scenario to mitigate these limitations is to deploy THz in clustered heterogeneous networks (HetNets).  In this setup, small base stations (SBS) and users are spatially clustered around traffic hotspots, allowing THz communication to focus on providing high-throughput and low-latency services within localized areas. Especially in real-world scenarios, users are not uniformly distributed but tend to cluster, forming concentrated groups.
This approach ensures that the unique strengths of THz technology are maximized while addressing its limitations through deployment strategies.  

There has been growing interest in leveraging THz bands to achieve ultra-high data rates. 
The work in \cite{10948140} considers a multi-tier hybrid sub-6GHz-mmWave-THz network, where each tier is modeled using PPP.
A system with both caching and THz transmissions has been explored in \cite{li2022collaborative}, where a two-tier cooperative caching strategy with base stations and device-to-device are modeled as PPP. Authors in \cite{sayehvand2020interference} consider a coexisting 
THz and radio frequency (RF) PPP-based HetNet, where the THz and the RF tiers are modeled as two independent PPPs. 
Motivated by these considerations, this work aims to explore the performance of THz communication in clustered HetNets.
The fundamental challenge of THz is considered in this work. 
Unlike the works in \cite{li2022collaborative} and \cite{sayehvand2020interference}, which typically
rely on PPP to model base stations and users in the THz context, we use
PCP to address the spatial limitations of THz communications.
In this work, we introduce a stochastic-geometry framework for THz clustered HetNets. We consider two-tier THz HetNets with MBSs tier modeled as a PPP, and SBSs modeled as PCP. The users are also modeled as PCP to consider the spatial correlation with the SBSs in the cluster hotspot. Given the THz vulnerability to blockages, we consider a 2D blockage effect by modeling
the buildings as a Boolean scheme of rectangles and distributing them using homogeneous PPP. Furthermore, we consider a sectored antenna model for all the nodes in the system to compensate for the high path loss associated with THz signals. We use Nakagami-$m$ fading to model the small-scale fading. Then we derive the analytical expressions for association probability, coverage probability, and the Laplace transform of the interference.  We validate the analytical results via Monte Carlo simulations and examine the effects of system parameters on network performance.

\vspace{-0.2cm}

\section{System model}
\vspace{-0.1cm}

We consider a downlink two-tier THz HetNet consisting of MBSs, SBSs, and users, as shown in Fig.~\ref{fig1}. The MBS tier is modeled as a homogeneous PPP $\Psi_m \subset \mathbb{R}^2$ with density $\lambda_m$. To capture hotspot deployment, the SBSs and users are spatially coupled and modeled as two separate PCPs sharing the same parent point process. Specifically, the cluster centers $\mathcal{D}_i$ form a homogeneous PPP $\Psi_p$ with density $\lambda_p$, independent of $\Psi_m$, while the SBSs and users are independently distributed as offspring points around these centers.
More specifically, the SBS tier follows a Thomas cluster process (TCP), a special case of isotropic and stationary PCP in which the offspring points are i.i.d. Gaussian distributed around the parent points \cite{ganti2009interference} \cite{obaid2026ris}. The offspring displacement has variance $\sigma^2$ and PDF $f(\boldsymbol{x})=\frac{1}{2\pi\sigma^2}\exp\big(-\frac{|\boldsymbol{x}|^2}{2\sigma^2}\big)$. Accordingly, the SBS point process is denoted by $\Psi_s=\mathcal{Q}(\Psi_p,\sigma_s,\bar{n}_s)$, where $\bar{n}_s$ is the average number of SBSs per cluster. Similarly, the user point process is $\Psi_u=\mathcal{Q}(\Psi_p,\sigma_u,\bar{n}_u)$, where $\bar{n}_u$ is the average number of users per cluster.
Let $k\in\{m,s\}$ denote the tier index, where $m$ and $s$ represent the MBS and SBS tiers, respectively. The reference user $\mathrm{U}$ is assumed to belong to the typical cluster $\mathcal{D}_{\mathbf{z}_0}$ centered at $\mathbf{z}_0\in\Psi_p$. Since both PPP and TCP are stationary, $\mathrm{U}$ is placed at the origin without loss of generality. Due to the sparse deployment of clusters, SBSs outside the typical cluster are relatively far from $\mathrm{U}$ and act only as interferers. Thus, $\mathrm{U}$ can associate with any MBS in the network or any SBS within the typical cluster. For tractability, the number of SBSs in the typical cluster is fixed and equal to $n_s$.

\vspace{-0.2cm}

\begin{figure}[h]
\centering
\includegraphics[width=0.27\textwidth]{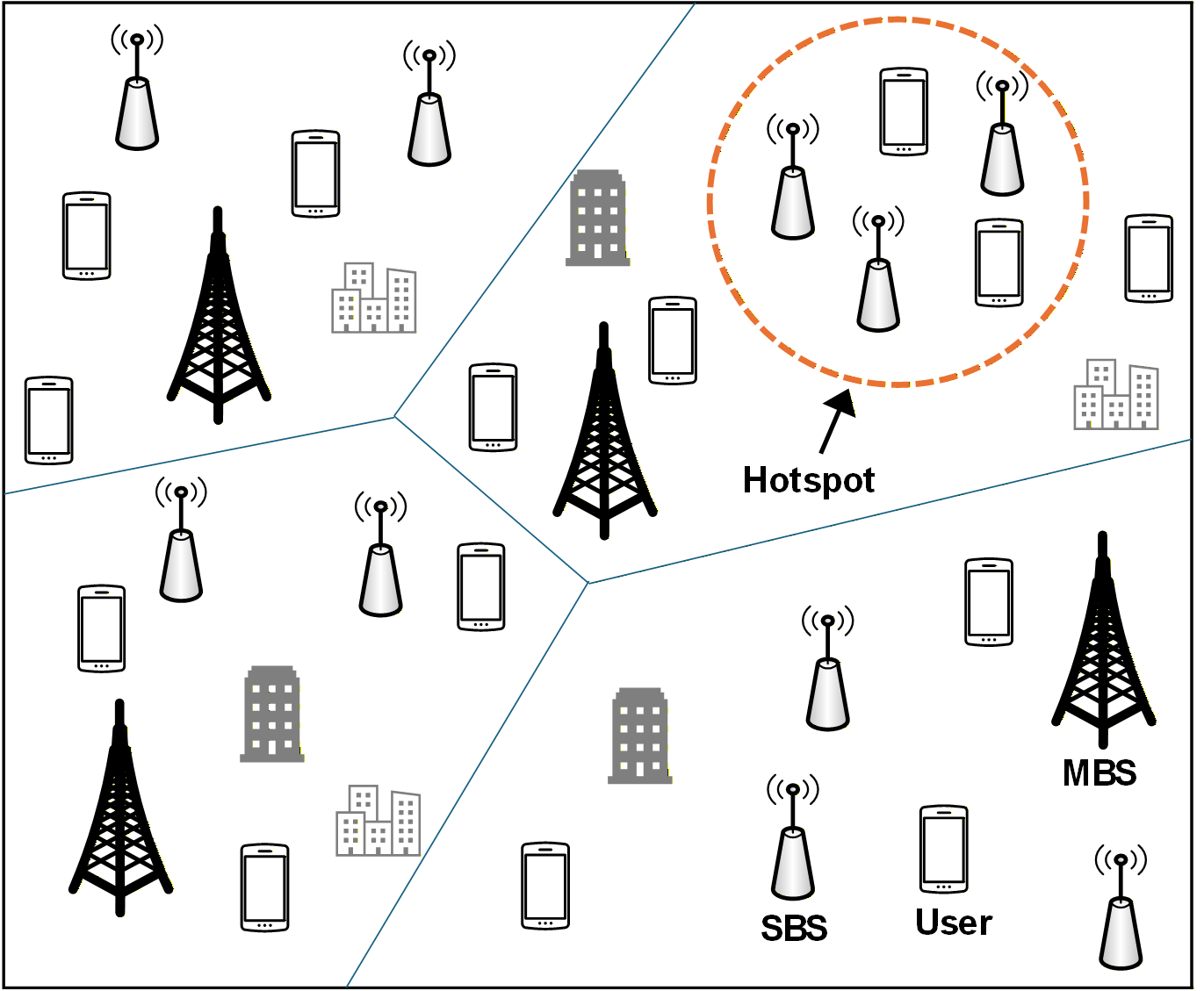}
\caption{ System model.} 
\label{fig1}
\vspace{-0.5cm}
\end{figure}

\subsection{ THz Channel Model}

\subsubsection{Antenna Model}
We adopt the sectored antenna model in \cite{bai2014coverage}. For a node of type $y\in\{u,m,s\}$, the antenna gain is given by
\vspace{-0.2cm}
\begin{align}
G_y(\Theta)=
\begin{cases}
G_y, & |\Theta|\leq b_y,\\
g_y, & |\Theta|> b_y,
\end{cases}
\vspace{-0.2cm}
\end{align}
where $\Theta\in(-\pi,\pi)$ denotes the boresight angle, while $G_y$, $g_y$, and $b_y$ are the main-lobe gain, side-lobe gain, and main-lobe beamwidth, respectively. The reference user and its serving BS align their main lobes; hence, only the main-lobe gains are considered for the desired link between $\text{U}$ and its serving base station. 
For interfering BSs, the antenna beams are not necessarily directed toward the reference user.
Thus, the directionality gain $G_{t_{i}}$ 
between the $\text{U}$ and an interferer is determined as a random variable $G_a$, $(a\in \{ 1,2,3,4\})$  which can take values from
$\{ {G}_{u}{{G}_{s/m}}$, ${G}_{u}{{g}_{s/m}}$, ${g}_{u}{{G}_{s/m}}$, ${g}_{u}{{g}_{s/m}} \}$, and probability $p_a$ for each case is $\{ \frac{\Theta_{u}}{2\pi}\frac{\Theta_{s/m}}{2\pi}$, $\frac{\Theta_{u}}{2\pi}(1-\frac{\Theta_{s/m}}{2\pi})$, $(1-\frac{\Theta_{u}}{2\pi})\frac{\Theta_{s/m}}{2\pi}$, $(1-\frac{\Theta_{u}}{2\pi})(1-\frac{\Theta_{s/m}}{2\pi})$\}.

\subsubsection{Blockage Model}
THz communications are highly affected by environmental blockages. In this paper, \textcolor{black}{ a Boolean scheme of rectangles is adopted to model the buildings} \cite{chiu2013stochastic}. The buildings are distributed as homogeneous PPP $ \Upsilon$ with density $ \lambda_B$, where the buildings' locations, sizes, and orientations are independent. The buildings' length $L_B$ and width $W_B$ are assumed to be i.i.d. uniformly distributed \cite{bai2014analysis} with mean values $\mathbb{E}{[L_B]}$ and $\mathbb{E}{[W_B]}$.

\subsubsection{Propagation Model}

THz communication is significantly affected by the loss of molecular absorption of water molecules; consequently, large-scale fading is represented through a deterministic exponential power loss propagation model \cite{obaid2024collaborative, obaid2024coverage}. THz non-line-of-sight (NLoS) links are obstructed by buildings and experience significant path loss, making them far less effective compared to line-of-sight (LoS) links. The notable power gap between LoS and NLoS paths makes THz channels LoS-dominant 
 \cite{wang2022joint}. 
 Thus, we consider only connections to the LoS MBSs and LoS SBSs. 
We use Nakagami-m fading to model the small-scale
fading with parameter $m$. Let $h_m$ and $h_s$ be the small-scale fading of the MBSs and SBSs links, respectively, and follow a Gamma distribution with shape and scale parameters $(m,1/m)$. 

The received signal power by $\text{U}$ from the MBSs and the SBSs links are expressed as $\zeta_{m} = P_m {G}_{m} {G}_{u} {h_m} \ell x_{m}^{-\alpha}e^{-K(f) x_{m}}$ and $\zeta_{s} = P_s {G}_{s} {G}_{u} {h_s}\ell x_{s}^{-\alpha}e^{-K(f) x_{s}}$, respectively.
$P_m$ and $P_s$ are the transmit power at MBSs and SBSs, respectively; ${G}_{m}$, ${G}_{s}$, and ${G}_{u}$ are the antenna's main lobe gain at the MBS, SBS, and $\text{U}$, respectively. 
$\alpha$ is the LoS path loss exponent; $K({f})$ is the molecular absorption coefficient of the transmission medium; ${\ell}={( {\frac{c}{{4\pi {f}}}} )^2}$ is the frequency-dependent coefficient; $c$ is the speed of light; and $f$ is the operating frequency. 

\vspace{-0.2cm}

\section{ Performance Analysis}
To derive the coverage probability, we need to derive relevant distance distributions and association probabilities. 
 \vspace{-0cm}
\subsection{The Probability of Establishing a LoS Link}
In THz communication, the links typically operate over shorter distances and are more susceptible to blockage effects. Based on the blockage model described in Section II, the LoS and NLoS probabilities of the transmission link from a base station to the reference user $\text{U}$ are given by 
\begin{align}
{P_L}{(r)}= \exp{\left(- (\beta r ) \right)},~~~~~~~ {P_N}{(r)}= 1- {P_L}{(r)},
\end{align} 
 where $\beta= [{2 \lambda_B (\mathbb{E}{[L_B]}+\mathbb{E}{[W_B]})}]/{\pi}$.

\subsection{User Association Policy}

We consider an association strategy that considers channel condition. Given that the THz channels are LoS dominant, the user will only associate with a LoS MBS or a LoS SBS. 
The user associates with the base station which provides the maximum received power, i.e. the serving base station is given by $x_{k}^* = \mathop {\arg \max }_{x\in \Psi_{m} \cup ~\mathcal{D}_{\mathbf{z_{0}}}}  \zeta_{k}, \quad k\in \left \{{ s,m }\right \}$.

The downlink received signal-to-noise-plus-interference ratio (SINR) can be expressed as
\vspace{-0.5cm}

\begin{align}\label{SINR}
    \gamma _{{k}}= \frac{\zeta_{k}}{ N_{o}+ I_{k}},
\end{align}
where $ k\in \left \{{ m,s}\right \}$, $ N_{o}$ is the thermal noise.
Here, ${I_{k}}$ denotes the aggregate interference. The interference experienced by $\text{U}$ in each association case can be defined as below:

1) When the reference user $\text{U}$ is served by an MBS, the interference encountered by the reference user originates from three sources: interference from the MBSs except the serving MBS $I_{mm} =\sum _{\boldsymbol {x_m}\in \Psi_m \backslash \boldsymbol {x}_{m}^{*}}   P_m G_{t_{i}} {h_m} \ell x_{m}^{-\alpha}e^{-K(f) x_{m}}$, intra-cluster interference from the SBSs within the cluster of the reference user $I^{\rm{Intra}}_{sm}=\sum _{\boldsymbol {x_s}\in \mathcal{D}_{\mathbf{z_{0}}}}  P_s G_{t_{i}} {h_s} \ell x_{s}^{-\alpha}e^{-K(f) x_{s}}$,
and inter-cluster interference from the SBSs outside the cluster of the reference user $I^{\rm{Inter}}_{sm}=\sum _{ {\mathcal{D}_{i}}\in \Psi_p \backslash \mathcal{D}_{\mathbf{z_{0}}} } \sum _{\boldsymbol {x_s}\in \mathcal{D}_{i} } P_s G_{t_{i}} {g_s} \ell x_{s}^{-\alpha}e^{-K(f) x_{s}}$.

2) When the reference user $\text{U}$ is served by an SBS from the SBS tier, the total interference experienced by the reference user arises from three sources: interference from the MBSs in the MBS tier $I_{ms}=\sum _{\boldsymbol {x_m}\in \Psi_m }  P_m G_{t_{i}} {h_m} \ell x_{m}^{-\alpha}e^{-K(f) x_{m}}$, intra-cluster interference from the SBSs within the cluster of the reference user except the serving SBS $I^{\rm{Intra}}_{ss}=\sum _{\boldsymbol {x}\in \mathcal{D}_{\mathbf{z_{0}}} \backslash \boldsymbol {x}_{s}^{*}}  P_s G_{t_{i}} {h_s} \ell x_{s}^{-\alpha}e^{-K(f) x_{s}}$, 
,
and inter-cluster interference from the SBSs outside the cluster of the reference user $I^{\rm{Inter}}_{ss}=\sum _{ {\mathcal{D}_{i}}\in \Psi_p \backslash \mathcal{D}_{\mathbf{z_{0}}} } \sum _{\boldsymbol {x_s}\in \mathcal{D}_{i} }  P_s G_{t_{i}} {h_s} \ell x_{s}^{-\alpha}e^{-K(f) x_{s}}$.

\subsection{Distance Distributions}
In this subsection, we will derive the distance distributions that are needed to derive the coverage probability.

\subsubsection{MBS Tier}
Let $r_{m}$ represents the distance corresponding to the potential nearest serving LoS MBS $x^{*}_{m}$ in $\mathbb{R}^{2}$. 
The cumulative distribution function (CDF) and the PDF of $r_{m}$ are denoted as $F_{{r_{m}}}(r)$ and $f_{{r_{m}}}(r)$, respectively. Then, the null probability of a homogeneous PPP is used to derive the distribution of $r_{m}$ as in \cite{bai2014coverage}
\vspace{-0.2cm}
\begin{align}
F_{{r_{m}}}(r)&=1-
\exp  (-2\pi \lambda_m \int _{0}^{r}{P_L}{(t)} t dt),
\end{align}
\vspace{-0.8cm}

\begin{align}
f_{{r_{m}}}(r)&=2\pi r  \lambda_m {P_L}{(r)} 
 \exp (-2\pi  \lambda_m \int _{0}^{r}{P_L}{(t)} t dt),
\end{align}


\color{red}

\color{black}
\subsubsection{SBS Tier}
In this tier, the SBSs are modeled as TCP. The CDF and the PDF of the conditional distance distribution for the distance of a randomly selected point in the SBS tier with a cluster center $\mathbf{z}_{0}$ are given according to \cite{afshang2018poisson}
 \vspace{-0.2cm}
\begin{align}
 F_{d_{s}}(x|{z_0})=1-Q_{1}\left({\frac {z_0}{\sigma_s}, \frac {x}{\sigma_s }}\right), \quad x>0, 
\end{align}
\vspace{-0.6cm}

\begin{align}
f_{d_{s}}(x|{z_0})= \frac {x}{\sigma_{s} ^{2}} \exp \left ({-\frac {x^{2}+z_0^{2}}{2 \sigma_{s} ^{2}}}\right) I_{0}\left ({\frac {x z_0}{\sigma_{s} ^{2}}}\right),\quad x> 0,
\end{align}
where $Q_1(a, b)$ is the Marcum $Q$-function $Q_1(a, b) = \int_{b}^{\infty} t e^{-t^2 - \frac{a^2}{2}} I_0(at) \, dt$, and $I_0(t) = \frac{1}{\pi}\int_{0}^{\pi} e^{t \cos \theta} \, d\theta$.

\begin{lemma}
In the SBS tier, conditioned on its cluster center $\mathbf{z_{0}}$, let $r_{s}$ be the distance of the potential nearest serving SBS $x^{*}_{s}$ in LoS condition. Given  $\mathbf{z_{0}}= z_{0} $, the CDF and the PDF of $r_s$ can be derived as
\vspace{-0.2cm}
  \begin{align}
  F_{r_{s}}(r|z_{0})= 1 - \left [{ 1-F_{c_{s}}(r|z_{0})} \right]^{ n_{s} },
  \end{align}
  \vspace{-0.4cm}
  \begin{align}
 f_{r_{s}}(r|z_{0})= n_s \left [{ 1-F_{c_{s}}(r|z_{0}) }\right]^{ n_s -1} f_{c_{s}}(r|z_{0}).
  \end{align}
\vspace{-0.4cm}
\begin{align} F_{c_{s}}(r|z_{0}) = \!\! \int _{0}^{r} \!\!\!\frac {y}{2\pi \sigma_{s} ^{2}} \exp \left ({-\frac {y^{2}+z_{0}^{2}}{2 \sigma_{s} ^{2}} }\right) {P_L}{(y)} I_0\left ({\frac {z_{0}y}{\sigma_{s} ^{2}} }\right) \mathrm {d}y,
\end{align}
\vspace{-0.4cm}
\begin{align} f_{c_{s}}(r|z_{0})= \frac { {P_L}{(r)} r}{2\pi \sigma_{s} ^{2}} \exp {\left ({-\frac {r^{2}+z_{0}^{2}}{2\sigma_{s} ^{2}} }\right)}  I_0\left ({\frac {z_{0} r}{\sigma_{s} ^{2}}}\right),\end{align}   

\end{lemma}

\begin{IEEEproof}
    Let $c_s$ be the distance from $\text{U}$ to a randomly chosen SBS in cluster $\mathcal{D}_{\mathbf{z_{0}}}$, the CDF of $c_s$ is derived as 
     \vspace{-0.2cm}
 \begin{align}\label{cdfs}
F_{c_s}(r|{z_0}) =&1 - \mathbb {P}\left [{ \text {No SBSs in $b(0, r)$} }\right] \nonumber \\
=&1-{\int _{ \mathbb{R}^{2}\backslash {b(0, r)}} f_{d_{s}}(x|{z_0}) {P_L}{(x)}  \,\mathrm {d}\boldsymbol {x}} \nonumber \\
=&{\int _{b(0, r)} f_{d_{s}}(x|{z_0}) {P_L}{(x)}\,\mathrm {d}\boldsymbol {x}} \nonumber \\
=&{\int _{0}^{r} f_{d_{s}}(y|{z_0}) {P_L}{(y)}\,\mathrm {d}{y}} ,
\end{align}
the PDF of $c_s$ can be obtained by differentiating the CDF of $c_s$ with respect to $r$. Then $ F_{r_{s}}(r|{z_0})= 1 - \left [{ 1-F_{c_{s}}(r|{z_0})} \right]^{ n_s }$ and the PDF of $r_{s}$ can be obtained by taking the derivative of $ F_{r_{s}}(r|{z_0})$ with respect to $r$.
\end{IEEEproof}


\vspace{-0.1cm}

\subsection{Association Probability}

In this subsection, we will derive the association probabilities, which represent the probability that the reference user is connected to a base station.

\begin{lemma}
The probability that $\text{U}$ is associated with an MBS is denoted as $\mathcal A_{m}(z_0)$ and can be derived as 
 \vspace{-0.1cm}
\begin{align}\label{Ass_m}
  \mathcal A_{m}(z_0)= &  \int _{0}^{\infty } (1- F_{c_{s}}(\delta_{m,s}(r|z_0))  )^{n_s} ~2\pi r  \lambda_m {P_L}{(r)}\nonumber \\
 & \times  
 \exp \bigg (-2\pi  \lambda_m \int _{0}^{r}{P_L}{(t)} t dt\bigg) dr,
\end{align}
where $\delta_{m,s}(r)= \frac{\alpha}{K(f)} {\cal W}\left(\frac{K(f)}{\alpha}  \left(\frac{G_s  P_s  e^{- K(f) r }}{G_m  P_m  r^{\alpha}}\right)^{\frac{1}{\alpha}}\right) $.

The probability that $\text{U}$ is associated with an SBS from the SBS tier is denoted as $\mathcal A_{s}(z_0)$ and can be obtained as
 \vspace{-0.1cm}
\begin{align}\label{Ass_s}
  \mathcal A_{s}(z_0) = &  \int _{0}^{\infty } n_s \left ({ 1-F_{c_{s}}(r|z_{0}) }\right)^{ n_s -1} f_{c_{s}}(r|z_{0}) \nonumber \\
 & \times {P_L}{(r)}   \exp  (-\pi  \lambda_m  \delta^{2}_{s,m}(r))dr, 
\end{align}

where $\delta_{s,m}(r)= \frac{\alpha}{K(f)} {\cal W}\left(\frac{K(f)}{\alpha}  \left(\frac{G_m P_m  e^{- K(f) r }}{G_s  P_s  r^{\alpha}}\right)^{\frac{1}{\alpha}}\right) $.
\end{lemma}

\begin{IEEEproof}
The reference user $ \text{U}$ associates with an MBS from the MBS tier if it is LoS and if $ \zeta_{m}  > \zeta_{s} $ i.e. $P_m {G}_{m} {G}_{U}  \ell r_{m}^{-\alpha}e^{-K(f) r_{m}}> P_s {G}_{s} {G}_{U} \ell r_{s}^{-\alpha}e^{-K(f) r_{s}}$. Thus, $\mathcal A_{m}(z_0)$ can be derived as
 \vspace{-0.1cm}
\begin{align} 
\mathcal A_{m}(z_0)&= 
 \mathbb {E}_{r_{m}}\bigg[ \mathbb {P}[\zeta_{m} > \zeta_{s}  ]\bigg ]   \! \!
= \! \! \int _{0}^{\infty }  \! \! \!\mathbb {P}\bigg [r_{s}>\delta_{m,s}(r) \bigg ]f_{{r_{m}}}(r) dr \nonumber \\
&=\int _{0}^{\infty } (1- F_{c_{s}}(\delta_{m,s}(r|z_0))  )^{n_s} ~2\pi r  \lambda_m {P_L}{(r)}\nonumber \\
 & \times  
 \exp \bigg (-2\pi  \lambda_m \int _{0}^{r}{P_L}{(t)} t dt\bigg) dr,
  \vspace{-0.3cm}
\end{align}
$\delta_{m,s}(r)$ and $\delta_{s,m}(r)$
are the distance thresholds and can be determined by solving the corresponding signal power equations $\zeta_m$ and $\zeta_s$, and  $\mathcal W(\cdot)$ is the principal Lambert-$W$ function.
A similar approach can be followed to derive $\mathcal A_{s}(z_0)$; hence, it is omitted.
\end{IEEEproof}


\begin{lemma}
    Let $X_{m}$ and $X_{s}$ be the serving distance from the MBS tier and SBS tier, respectively. Then, the PDF of $X_{m}$ and $X_{s}$ are given as 
     \vspace{-0.2cm}
    \begin{align}\label{Xm}
    f_{X_{m}}\left ({x|z_{0}}\right)=
    &\frac {f_{{r_{m}}}(x)}{ \mathcal A_{m}(z_0)} (1-  F_{r_{s}}(\delta_{m,s}(x)|(z_0)), 
    \end{align}
    \vspace{-0.7cm}

      \begin{align}
    f_{X_{s}}\left ({x|z_{0}}\right)=
    &\frac {f_{{r_{s}}}(x|z_{0}) }{\mathcal A_{s}(z_0)} (1-  F_{r_{m}}(\delta_{s,m}(x)).
    \end{align}
\end{lemma}


\begin{IEEEproof}
The proof is similar to  Lemma 3 in \cite{jo2012heterogeneous}.
\end{IEEEproof}



 \vspace{-0.02cm}
\section{Coverage Probability}

 \vspace{-0.1cm}
The coverage probability $P_C$ can be defined as the probability that the SINR  at the user exceeds a predefined threshold $T$, i.e., $P_C$ is given by $\mathbb {P}(\gamma \geq T)$.
Then the downlink coverage probability can be expressed as
\vspace{-0.2cm}
\begin{align}
    P_C(T)= &\mathbb {E}_{\mathbf{z_{0}}}\bigg[  \mathcal A_{m}(z_0) P_{C,m}(T|z_0)  + \mathcal A_{s}(z_0)P_{C,s}(T|z_0)\bigg],
    \vspace{-0.2cm}
\end{align}
where $P_{C,m}(T|z_0)$ and $P_{C,s}(T|z_0)$ are the corresponding conditional coverage probabilities. After deriving $P_{C,m}(T|z_0)$ and $P_{C,s}(T|z_0)$ along with the  corresponding Laplace transform of the interference,  we can obtain the final result of $P_C$ by taking the expectation of
$P_C$ with respect to $z_{0}$ and evaluating the integral. Here $ f_{\mathbf{z_{0}}}\left ({z_{0}}\right)=\frac {z_{0}}{\sigma _ {u} ^{2}} \exp \left ({-\frac {z_{0}^{2}}{2\sigma _ {u} ^{2}} }\right)$.

\begin{lemma}
The conditional coverage probability that $\text{U}$ is served by an MBS  can be expressed as 
 \vspace{-0.1cm}
\begin{align}\label{pc_m}
P_{C,m}(T|z_0) =& \!\!\int _{0}^\infty \bigg [  \sum _{n=1}^{m} \left ({-1}\right)^{n+1} \binom {m}{n} \exp \left (\!\! -\frac { n \Lambda T N_{o} }{s_1}\!\!\right) \nonumber \\
&\times\mathcal {L}_{I_{m}}\!\!\left (\!{\frac {n\Lambda T  }{s_1}|z_{0},x}\right)  \bigg]\, f_{X_{m}}\left ({x|z_{0}}\right) \mathrm {d}x,
\end{align}
 \vspace{-0.9cm}
 
\begin{align} \mathcal {L}_{I_{m}}\left ({s_1|z_{0},x}\right)&= \mathcal {L}_{I_{mm}}\left ({s_1|z_{0},x}\right) \cdot \mathcal {L}_{I^{\rm{Intra}}_{sm}}\left ({s_1|z_{0},x}\right)\nonumber \\
&\cdot \mathcal {L}_{I^{\rm{Inter}}_{sm}}\left ({s_1|z_{0},0}\right),
\end{align}
 \vspace{-0.8cm}
\begin{align} 
&\mathcal {L}_{I_{mm}}\left ({s_1|z_{0},x}\right)= \exp \Big ({ - 2\pi \lambda_{m}} \times \, \int _{x}^\infty \sum_{a=1}^{4} {p_a}\nonumber \\
&\times{\left (1-{\frac{1}{(1+s_1 P_m G_a  \ell w^{-\alpha}e^{-K(f) w})^{m}  }}\right) } w \,\mathrm {d}w\Big ),
\end{align}
 \vspace{-0.3cm}
\begin{align}
&\mathcal {L}_{I^{\rm{Intra}}_{sm}}\left ({s_1|z_{0},x}\right)= \exp  \!\Big [ \!- 2  \pi n_s \int _{ \delta_{m,s}(x)}^{\infty }  \! \!f_{X_{s}}\left ({w|z_{0}}\right)   \nonumber \\
&\times \sum_{a=1}^{4} {p_a} {\left ( \!1-{\frac{1}{(1+s_1 P_s G_a  \ell w^{-\alpha}e^{-K(f) w})^{m}  }}  \!\right) }  w {\mathrm{ d}}w \Big ],
\end{align}
 \vspace{-0.4cm}
\begin{align}
\mathcal {L}_{I^{\rm{Inter}}_{sm}}\left ({s_1|z_{0},0}\right)&= \exp \Big [- 2  \pi \lambda_{p}     \nonumber \\
& \times\int _{ 0}^{\infty }\left [{ 1 - \mathcal {L}_{{I_{sm}^{\rm{intra}}} }\left ({s_1|z,0}\right)} \right] z {\mathrm{ d}}z \Big ],
\end{align}
where $\mathcal {L}_{I_{m}}(.)$ is the Laplace transform (LT) of the interference when the reference user associates with the serving MBS, and $s_1 = P_m {G}_{m} {G}_{u}  \ell x_{m}^{-\alpha}e^{-K(f) x_{m}}$, $\Lambda= m( m!)^{-1/m}$.

\end{lemma}

\begin{IEEEproof}
   The MBS coverage probability is given by
    \vspace{-0.25cm}
\begin{align}
&P_{C,m}(T|z_0)\nonumber \\
&=\mathbb {P}\left [{ \left.{ \frac {P_m {G}_{m} {G}_{u} {h_m} \ell x_{m}^{-\alpha}e^{-K(f) x_{m}} }{N_{o}+ I_{m}} > T}\right | \mathbf{z_{0}}=z_{0} }\right] \nonumber \\
& \overset{(a)}{=}\mathbb {P}\left [{ \left.{ h_m > \frac {T  ({ I_{m}+N_{o}})}{P_m {G}_{m} {G}_{u}  \ell x_{m}^{-\alpha}e^{-K(f) x_{m}}} }\right |\mathbf{z_{0}}=z_{0} }\right] \nonumber \\
& \overset{(b)}{=}\mathbb {E}_{x_m,I_{m}} \bigg[ \sum _{n=1}^{m} \left ({-1}\right)^{n+1} \binom {m}{n}\nonumber \\  
&\quad \times \exp \bigg (-\frac {T  ({ I_{m}+N_{o}})}{P_m {G}_{m} {G}_{u}  \ell x_{m}^{-\alpha}e^{-K(f) x_{m}}} \bigg )  \bigg] \nonumber \\
& \overset{(c)}{=}\int _{0}^\infty \bigg[ \sum _{n=1}^{m} \left ({-1}\right)^{n+1} \binom {m}{n}   \exp \left ( -\frac {n\Lambda T N_{o} }{s_1}\right) \nonumber \\
& \quad \times \mathcal {L}_{I_{m}}\left (\frac {n\Lambda T  }{s_1}\bigg|z_{0},x\right) \bigg] \, f_{X_{m}}\left (x\big|z_{0}\right) \mathrm {d}x,
 \vspace{-0.2cm}
\end{align}
where $(a)$ according to the complementary cumulative distribution function (CCDF) of $
h_m$, (b) according to the Alzer’s inequality \cite{alzer1997some} of $h_m$, and averaging over the independent random variables $x$ and $I_{m}$, $(c)$ from $s_1 $, $\Lambda$, and the LT definition of interference. 

 $\mathcal {L}_{I_{mm}}\left ({s_1|z_{0},x}\right)$  is the interference on the reference user from the MBSs except the serving MBS and it can be obtained following the same lines as \cite{andrews2011tractable}.
The derivation of $\mathcal {L}_{I^{\rm{Intra}}_{sm}}\left ({s_1|z_{0},x}\right)$ can be derived as follows
\begin{align} 
&\mathcal {L}_{I^{\rm{Intra}}_{sm}}\left ({s_1|z_{0},x}\right) \nonumber \\
&\stackrel {(a)}{=}\mathbb {E}_{ \mathcal{D}_{\mathbf{z_{0}}},G_{t},{h_s}}\bigg[{\exp \bigg({-s_1  \! \!  \! \!  {\sum _{\boldsymbol {x}\in \mathcal{D}_{\mathbf{z_{0}}}}}  \! \!  \! \!  P_s G_{t} {h_s} \ell x_{s}^{-\alpha}e^{-K(f) x_{s}}}\bigg)}\bigg]\nonumber \\
&\stackrel {(b)}{=}\mathbb {E}_{ \mathcal{D}_{\mathbf{z_{0}}},G_{t}}\bigg[~{\prod _{{\boldsymbol {x}\in \mathcal{D}_{\mathbf{z_{0}}}}} \frac {1}{(1+s_1 P_s G_{t}  \ell x_{s}^{-\alpha}e^{-K(f) x_{s}})^{m} }}\bigg] \nonumber \\
& \stackrel {(c)}{=}\exp \Big [- 2  \pi n_s \int _{ \delta_{m,s}(r)}^{\infty } f_{X_{s}}\left ({w|z_{0}}\right) \sum_{a=1}^{4} {p_a}   \nonumber \\
& \qquad \times\bigg[1-  \frac {1}{(1+s_1 P_s G_a  \ell w^{-\alpha}e^{-K(f) w})^{m} }\bigg]    w {\mathrm{ d}}w \Big ],
\end{align}
where (a) from the definition of LT, 
(b) the moment
generating functional of $h_s$, and (c) from the probability generating functional (PGFL) of Poisson process with intensity measure $  n_s f_{X_{s}}\left ({w|z_{0}}\right)$, and averaging over $G_{t}$.

 Then $\mathcal {L}_{I^{\rm{Inter}}_{sm}}\left ({s_1|z_{0},0}\right)$ can be derived as follows
\begin{align}
 & \mathcal {L}_{I^{\rm{Inter}}_{sm}}\left ({s_1|z_{0},0}\right)
 =\mathbb {E}_{\Psi _{p}}  \Big [{ \prod _{\mathcal{D}\in \Psi _{p}} \mathcal {L}_{\left.{I_{sm}^{{intra}}}\right |\mathcal{D}}\left ({s_1;\lVert \mathcal{D}\rVert,0}\right) } \Big] \nonumber \\
 \overset {(a)}{=}&\exp  \Big [{\! -2\pi \lambda _{p} \int _{0}^\infty \!\left ({\! 1 - \mathcal {L}_{{I_{sm}^ { {intra}}} }\left ({s_1;z,0}\right) \!}\right)\! z \,\mathrm {d}z \!} \Big],
  \end{align}
where (a) from the PGFL of $\Psi_p$. 
\end{IEEEproof}


\begin{lemma}
The conditional coverage probability that the reference user $\text{U}$ is served by an SBS is given by
\vspace{-0.4cm}

\begin{align}\label{pc_s}
P_{C,s}(T|z_0) =& \int _{0}^\infty \bigg [  \sum _{n=1}^{m} \left ({-1}\right)^{n+1} \binom {m}{n} \exp \left (\! -\frac {n\Lambda T N_{o} }{s_2}\!\right)\nonumber \\
&\times \mathcal {L}_{I_{s}}\!\left ({\frac {n\Lambda T  }{s_2}|z_{0},x}\right)\Big ] \,f_{X_{s}}\left ({x|z_{0}}\right) \mathrm {d}x,
\end{align}
\vspace{-0.8cm}

\begin{align} \mathcal {L}_{I_{s}}\left ({s_2|z_{0},x}\right)&= \mathcal {L}_{I_{ms}}\left ({s_2|z_{0},x}\right) \cdot \mathcal {L}_{I^{\rm{Intra}}_{ss}}\left ({s_2|z_{0},x}\right)\nonumber \\
&\cdot \mathcal {L}_{I^{\rm{Inter}}_{ss}}\left ({s_2|z_{0},0}\right),
\end{align}
where $\mathcal {L}_{I_{s}}(.)$ is the LT of the interference when the reference user associates with the serving SBS, and $s_2 = P_s {G}_{s} {G}_{u}  \ell x_{s}^{-\alpha}e^{-K(f) x_{s}}$. Where 
\vspace{-0.3cm}
\begin{align}
&\mathcal {L}_{I_{ms}}\left ({s_2|z_{0},x}\right)  = \exp \Big ({ - 2\pi \lambda_{m}} \times \, \int _{\delta_{s,m}(x)}^\infty \sum_{a=1}^{4} {p_a}\nonumber \\
&\times{\bigg[1-  \frac {1}{(1+s_2 P_m G_a  \ell w^{-\alpha}e^{-K(f) w})^{m} }\bigg]  } w \,\mathrm {d}w \Big ),
\end{align}
 \vspace{-0.3cm}
\begin{align}
&\mathcal {L}_{I^{\rm{Intra}}_{ss}}\left ({s_2|z_{0},x}\right) = \exp \Big [- 2  \pi ( n_s-1) \int _{ x}^{\infty } f_{X_{s}}\left ({w|z_{0}}\right)  \nonumber \\
&  \times\sum_{a=1}^{4} {p_a} \bigg[1-  \frac {1}{(1+s_2 P_s G_a  \ell w^{-\alpha}e^{-K(f) w})^{m} }\bigg]   w {\mathrm{ d}}w \Big ],
\end{align}
\vspace{-0.4cm}

\begin{align}
\vspace{-0.4cm}
\mathcal {L}_{I^{\rm{Inter}}_{ss}}&\left ({s_2|z_{0},0}\right)= \exp \Big [- 2  \pi \lambda_{p}     \nonumber \\
& \times\int _{ 0}^{\infty }\left [{ 1 - \mathcal {L}_{\left.{I_{ss}^{{intra}}}\right |n_s +1}\left ({s_2|z,0}\right)} \right] z {\mathrm{ d}}z \Big ],
\end{align}

\end{lemma}

\begin{IEEEproof}
   The derivation of $P_{C,s}(T|z_0)$ can be obtained following similar procedures as $P_{C,m}(T|z_0)$. 
\vspace{-0.2cm}
\begin{align}
&\mathcal {L}_{I^{\rm{Intra}}_{ss}}\left ({s_2|z_{0},x}\right) \nonumber \\
&\stackrel {(a)}{=}\mathbb {E}_{ \mathcal{D}_{\mathbf{z_{0}}},G_{t},{h_s}}\bigg[{\exp \bigg({-s_2 \! \!\! \! \!  \! \!  { \sum _{\boldsymbol {x}\in \mathcal{D}_{\mathbf{z_{0}}}\backslash x_{s}^*} } \! \! \! \!  \! \!  \! \!  P_s G_{t} {h_s} \ell x_{s}^{-\alpha}e^{-K(f) x_{s}}}\bigg)}\bigg]\nonumber \\
&\stackrel {(b)}{=}\mathbb {E}_{ \mathcal{D}_{\mathbf{z_{0}}},G_{t}}\bigg[~{\prod _{{\boldsymbol {x}\in \mathcal{D}_{\mathbf{z_{0}}}} \backslash x_{s}^* }  \frac {1}{(1+s_2 P_s G_{t}  \ell x_{s}^{-\alpha}e^{-K(f) x_{s}})^{m} }}\bigg] \nonumber \\
&\stackrel {(c)}{=}\exp \bigg [- 2  \pi ( n_s-1) \int _{ x}^{\infty } f_{X_{s}}\left ({w|z_{0}}\right)  \sum_{a=1}^{4} {p_a}  \nonumber \\
& \quad \times \bigg[1-  \frac {1}{(1+s_2 P_s G_a  \ell w^{-\alpha}e^{-K(f) w})^{m} }\bigg]  w {\mathrm{ d}}w \bigg],
\end{align}
where (a) from the definition of LT, 
(b) according to the moment
generating functional of $h_s$, and (c) from the PGFL of Poisson process with intensity measure $ (n_s -1)f_{X_{s}}\left ({w|z_{0}}\right)$, and averaging over $G_{t}$.
The derivation of $\mathcal {L}_{I_{ms}}\left ({s_2|z_{0},x}\right) $ can be derived following same lines as \cite{jo2012heterogeneous} with minor modification. The derivation of $\mathcal {L}_{I^{\rm{Inter}}_{ss}}\left ({s_2|z_{0},0}\right)$ can be derived following similar approach as $\mathcal {L}_{I^{\rm{Inter}}_{sm}}\left ({s_1|z_{0},0}\right)$.
\end{IEEEproof}


\begin{theorem}
 The downlink coverage probability is given by
 \vspace{-0.5cm}

\begin{align}\label{total_pc}
P_C(T)&= \int _{ 0}^{\infty }  f_{\mathbf{z_{0}}}\left ({z_{0}}\right) \bigg[ \mathcal A_{m}(z_0)  \int _{0}^\infty  \bigg [  \sum _{n=1}^{m} \left ({-1}\right)^{n+1} \nonumber \\
& \!\!\times \binom {m}{n}\exp \left (\!\! -\frac {  n\Lambda T N_{o} }{s_1}\!\!\right) 
\mathcal {L}_{I_{m}}\left (\!{\frac {n\Lambda T  }{s_1}|z_{0},x}\right)  \bigg]  \nonumber \\
&\!\!\times f_{X_{m}}\left ({x|z_{0}}\right) \mathrm {d}x+  \mathcal A_{s}(z_0) \!\! \int _{0}^\infty \! \bigg [  \sum _{n=1}^{m} \left ({-1}\right)^{n+1}  \nonumber \\
&\!\!\times \binom {m}{n} \exp \left (\! -\frac {n\Lambda T N_{o} }{s_2}\!\right)\mathcal {L}_{I_{s}}\!\left ({\frac {n\Lambda T  }{s_2}|z_{0},x}\right)\bigg] \nonumber \\
&\!\!\times f_{X_{s}}\left ({x|z_{0}}\right) \mathrm {d}x  \bigg] \mathrm {d}z_{0}.
\end{align}   
\end{theorem}
\vspace{-0.2cm}

\begin{IEEEproof}
    The results are obtained by taking the expectation of $P_C(T)$ with respect to $\mathbf{z_{0}}$ then plugging in \eqref{pc_m} and \eqref{pc_s}.
\end{IEEEproof}

\begin{figure*}[!t]
\centering

\begin{subfigure}[t]{0.26\textwidth}
    \centering
    \includegraphics[width=\linewidth]{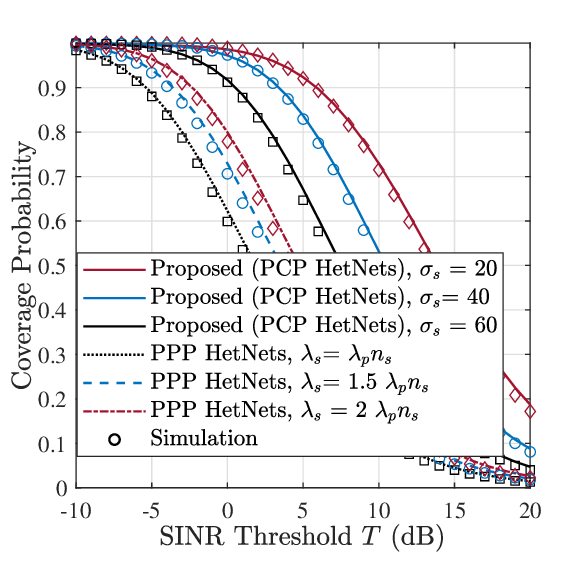}
    \caption{Coverage versus $T$.}
    \label{CP_SINR}
\end{subfigure}\hspace{-0.4cm}
\begin{subfigure}[t]{0.26\textwidth}
    \centering
    \includegraphics[width=\linewidth]{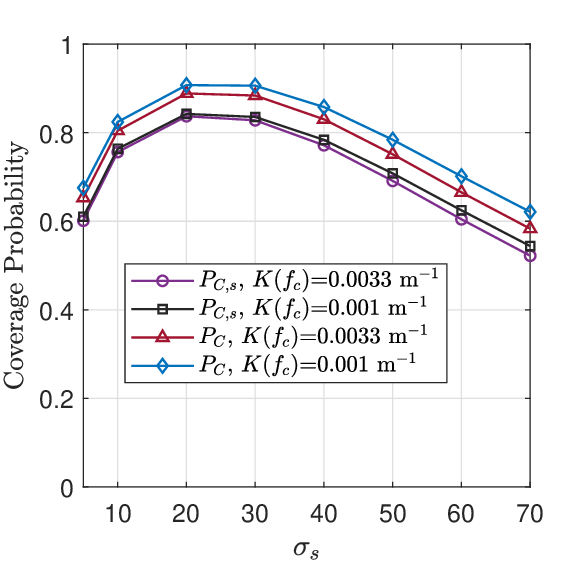}
    \caption{Coverage versus $\sigma_s$.}
    \label{CP_sigma_s}
\end{subfigure}\hspace{-0.4cm}
\begin{subfigure}[t]{0.26\textwidth}
    \centering
    \includegraphics[width=\linewidth]{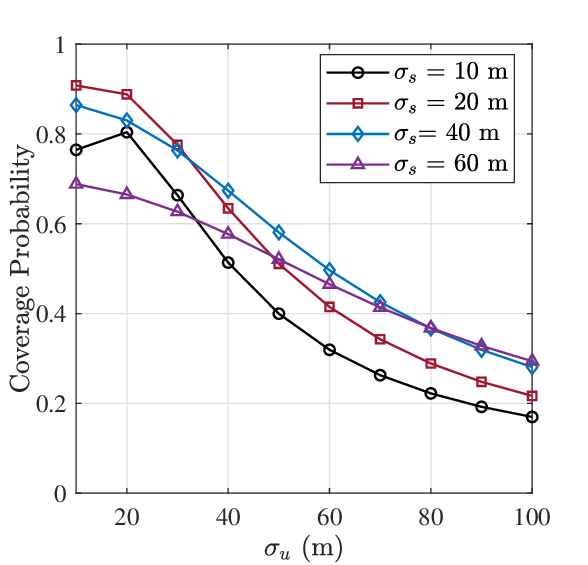}
    \caption{Coverage versus $\sigma_u$.}
    \label{CP_sigma_u}
\end{subfigure}\hspace{-0.4cm}
\begin{subfigure}[t]{0.26\textwidth}
    \centering
    \includegraphics[width=\linewidth]{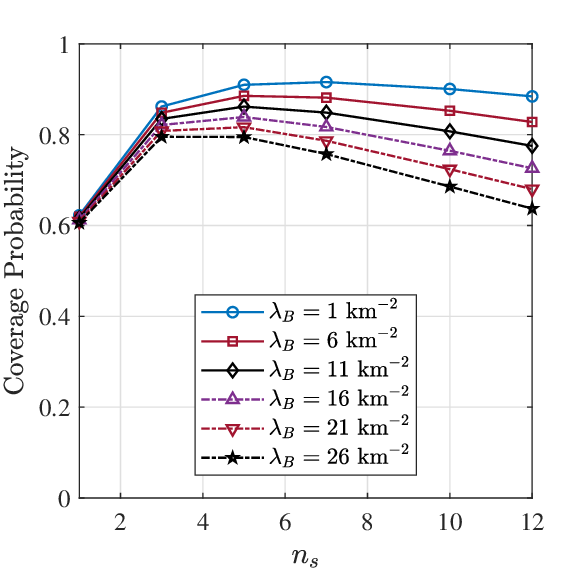}
    \caption{Coverage versus $n_s$.}
    \label{CP_ns}
\end{subfigure}
\vspace{-0.17cm}

\caption{Simulation results for coverage probability.}
\label{fig:coverage_four_results}
\vspace{-0.35cm}
\end{figure*}

\section{Numerical Results and Discussion}
In this section, we evaluate the system's performance through simulation and numerical results.  The system parameters are set as follows: $\alpha=2$, $G_U=10$ dBi, $g_u=-10$ dBi, $G_m=25$ dBi, $g_m=-10$ dBi, $G_s=25$ dBi, $g_s=-10$ dBi, $\Theta_{s/m}=10^\circ$, $\Theta_u=33^\circ$, $K(f)=0.0033$ m$^{-1}$, $f=300$ GHz, $P_m=20$ dBm, $P_s=10$ dBm, $T=5$ dB, $\lambda_m=5/\text{km}^2$, $\lambda_p=7/\text{km}^2$, $n_s=7$, $\sigma_s=20$, $\sigma_u=20$, $m=3$, $N_o=4\times10^{-11}$ Watt, $\mathbb{E}[L_B]=20$ m, $\mathbb{E}[W_B]=15$ m, and $\lambda_B=5/\text{km}^2$.
We validate the analytical results via Monte Carlo simulations with $10^5$ iterations.

In Fig. \ref{CP_SINR}, the analytical coverage is compared with Monte Carlo simulation results for various SINR threshold values. The analytical results align closely with the simulation results.
The figure also compares the PPP-based HetNet, where all the nodes (SBS tier, MBS tier, and user) are modeled as independent PPPs. To ensure a fair comparison with PCP HetNets, we set $\lambda_s $= $ \lambda_p n_s =7/   km^{2} \times 7$. It can be observed that the coverage in the PPP HetNet with THz is lower than that achieved by the proposed system. This occurs because in PPP-based HetNets, nodes are uniformly distributed, which increases the average serving distance and reduces the LoS probability at THz frequencies. In contrast, clustered SBS deployment via PCP brings users closer to serving nodes, shortening link distances, enhancing LoS, and improving signal quality.
The figure also shows an improvement in the coverage as $\sigma_s$ decreases. 
We further investigate the impact of $\sigma_s$ in Fig.~\ref{CP_sigma_s}. The results show that both the $P_C$ and $P_{C,s}$ exhibit a non-monotonic behavior with respect to $\sigma_s$. Specifically, both metrics first increase as $\sigma_s$ grows and attain their maximum at $\sigma_s =20$ m, and then gradually decrease as $\sigma_s$ increases further. This indicates that a moderate spatial spread of SBSs is beneficial for coverage. When $\sigma_s$ is very small, the SBSs are highly concentrated around the cluster center, which can intensify interference and limit the coverage. As $\sigma_s$ increases to a moderate value, the SBS layout becomes more favorable, improving the coverage performance. However, when $\sigma_s$ becomes large, the distance between the user and its serving SBS increases, which weakens the desired signal power and reduces the coverage. 
The figure also shows that a smaller  $K(f)=0.001~\mathrm{m}^{-1}$ results in higher coverage than $K(f)=0.0033~\mathrm{m}^{-1}$, since lower absorption reduces propagation loss and improves link reliability.

 Fig. \ref{CP_sigma_u} shows the coverage probability versus the user-cluster spread $\sigma_u$ for different SBS-cluster spreads $\sigma_s$. Overall, the coverage decreases as $\sigma_u$ increases, since users become more dispersed from the hotspot center, increasing the average serving distance and weakening the LoS THz links. For $\sigma_s=20$ m, a slight non-monotonic behavior is observed, where the coverage first improves and reaches its maximum around $\sigma_u=20$ m. This occurs because a moderate user spread reduces severe intra-cluster interference while still maintaining relatively short serving distances. However, for larger $\sigma_u$, the serving distance becomes dominant and the coverage decreases.
Fig. \ref{CP_ns} shows the coverage probability versus $n_s$, for different building densities $\lambda_B$. The coverage first increases with $n_s$ because more SBSs improve the chance of establishing a nearby LoS THz link. However, the gain saturates and may slightly decrease at larger $n_s$ due to stronger intra-cluster interference. 
Moreover, increasing $\lambda_B$ consistently reduces the coverage probability, as denser blockage environments lower the likelihood of maintaining LoS THz links and intensify the vulnerability of the network to signal obstruction.


\vspace{-0.2cm}

\section{Conclusion}
\vspace{-0.1cm}
In this paper, we investigated the performance of THz communication in clustered HetNets. Using stochastic geometry tools, we modeled the SBS tier and users using PCP and the MBS tier using PPP in a realistic deployment scenario. This framework mitigates the limitations of THz communication, such as short transmission ranges and high signal attenuation. We derived expressions for association probabilities and coverage probability, taking into account environmental blockages and unique characteristics of THz.
The numerical results show that the coverage in THz PCP-HetNets is higher than that achieved in THz PPP HetNets. In addition, a moderate spatial spread of SBSs is beneficial for coverage.
The insights from this study offer guidance for deploying THz communication systems in next-generation networks.

\vspace{-2mm}
\section*{Acknowledgment}
\vspace{-1mm}
This work was supported by Tekniikan edistämissäätiö (Technology
Promotion Foundation, Finland, and partially supported by the DIOR project (10100828) which has received funding from the European Union’s MSCA RISE programme.

\vspace{-0.2cm}

\bibliographystyle{IEEEtran}
\bibliography{bib}

@STRING{IEEE_J_WCOM       = "{IEEE} Trans. Wireless Commun."}

@STRING{IEEE_J_VT         = "{IEEE} Trans. Veh. Technol."}

@STRING{IEEE_J_IOT        = "{IEEE} Internet Things J."}

@ARTICLE{e,
  author={{Yang, Zhaohui and Xu, Wei and Shikh-Bahaei, Mohammad}},
  journal=IEEE_J_VT, 
  title={{Energy Efficient UAV Communication With Energy Harvesting}},
  month={Feb.},
  year={2020},
  volume={69},
  number={2},
  pages={1913-1927},
  doi={10.1109/TVT.2019.2961993}}

@ARTICLE{g,
  author={S. {Fu} and Y. {Tang} and Y. {Wu} and N. {Zhang} and H. {Gu} and C. {Chen} and M. {Liu}},
  journal=IEEE_J_IOT , 
  title={{Energy-Efficient UAV enabled Data Collection via Wireless Charging: A Reinforcement Learning Approach}}, 
  month={Jan.},
  year={2021},
  volume={},
  number={},
  pages={1-11},
  doi={10.1109/JIOT.2021.3051370}}

@ARTICLE{m,
  author={Y. {Zeng} and R. {Zhang}},
  journal=IEEE_J_WCOM,
  title={{Energy-Efficient UAV Communication With Trajectory Optimization}}, 
  month={June},
  year={2017},
  volume={16},
  number={6},
  pages={3747-3760},
  doi={10.1109/TWC.2017.2688328}}

@ARTICLE{1,
  author={N. I. {Mowla} and N. H. {Tran} and I. {Doh} and K. {Chae}},
  journal=IEEE_J_CN, 
  title={{AFRL: Adaptive Federated Reinforcement Learning for Intelligent Jamming Defense in FANET}}, 
  month={June},
  year={2020},
  volume={22},
  number={3},
  pages={244-258},
  doi={10.1109/JCN.2020.000015}}

@book{chiu2013stochastic,
  title={Stochastic geometry and its applications},
  author={Chiu, Sung Nok and Stoyan, Dietrich and Kendall, Wilfrid S and Mecke, Joseph},
  year={2013},
  publisher={John Wiley \& Sons}
}

@article{bai2014analysis,
  title={Analysis of blockage effects on urban cellular networks},
  author={Bai, Tianyang and Vaze, Rahul and Heath, Robert W},
  journal={IEEE Trans. Wireless Commun.},
  year={2014},
  publisher={IEEE}
}

@article{sayehvand2020interference,
  title={Interference and coverage analysis in coexisting RF and dense terahertz wireless networks},
  author={Sayehvand, Javad and Tabassum, Hina},
  journal={IEEE Wireless Commun. Lett.},
  volume={9},
  number={10},
  pages={1738--1742},
  year={2020},
  publisher={IEEE}
}

@article{wang2022joint,
  title={Joint hybrid 3D beamforming relying on sensor-based training for reconfigurable intelligent surface aided terahertz-based multiuser massive MIMO systems},
  author={Wang, Xufang and Lin, Zihuai and Lin, Feng and Hanzo, Lajos},
  journal={IEEE Sensors Journal},
  volume={22},
  number={14},
  pages={14540--14552},
  year={2022},
  publisher={IEEE}
}

@article{bai2014coverage,
  title={Coverage and rate analysis for millimeter-wave cellular networks},
  author={Bai, Tianyang and Heath, Robert W},
  journal={IEEE Trans. Wireless Commun.},
  year={2014},
  publisher={IEEE}
}

@article{afshang2018poisson,
  title={Poisson cluster process based analysis of HetNets with correlated user and base station locations},
  author={Afshang, Mehrnaz and Dhillon, Harpreet S},
  journal={IEEE Trans. Wireless Commun.},
  volume={17},
  number={4},
  pages={2417--2431},
  year={2018},
  publisher={IEEE}
}

@article{ganti2009interference,
  title={Interference and outage in clustered wireless ad hoc networks},
  author={Ganti, Radha Krishna and Haenggi, Martin},
  journal={IEEE Trans Inf Theory.},
  year={2009},
  publisher={IEEE}}

@article{andrews2011tractable,
  title={A tractable approach to coverage and rate in cellular networks},
  author={Andrews, Jeffrey G and Baccelli, Fran{\c{c}}ois and Ganti, Radha Krishna},
  journal={IEEE Transactions on communications},
  volume={59},
  number={11},
  pages={3122--3134},
  year={2011},
  publisher={IEEE}
}

@article{alzer1997some,
  title={On some inequalities for the incomplete gamma function},
  author={Alzer, Horst},
  journal={Mathematics of Computation},
  volume={66},
  number={218},
  pages={771--778},
  year={1997}
}

@article{li2022collaborative,
  title={A collaborative caching-transmission method for heterogeneous video services in cache-enabled terahertz heterogeneous networks},
  author={Li, Qi and Nayak, Amiya and Wang, Xiaoxiang and Wang, Dongyu and Yu, F Richard},
  journal={IEEE Trans. Veh. Technol.},
  publisher={IEEE}
}

@article{jo2012heterogeneous,
  title={Heterogeneous cellular networks with flexible cell association: A comprehensive downlink SINR analysis},
  author={Jo, Han-Shin and Sang, Young Jin and Xia, Ping and Andrews, Jeffrey G},
  journal={IEEE Trans. Wireless Commun.},
  year={2012},
  publisher={IEEE}
}

@article{obaid2024collaborative,
  title={Collaborative caching and reconfigurable intelligent surface for the sub-thz mobile system},
  author={Obaid, Hadeel and Zhu, Yongxu and Tan, Bo},
  journal={IEEE Wirel. Commun. Lett.},
  year={2024},
  publisher={IEEE}
}

@inproceedings{obaid2024coverage,
  title={Coverage Analysis of a THz Aerial Base Station Wireless Network in a Finite Area},
  author={Obaid, Hadeel and Zhu, Yongxu and Tan, Bo},
  booktitle={2024 IEEE 100th VTC (VTC2024-Fall)},
  pages={1--5},
  year={2024},
  organization={IEEE}
}

@ARTICLE{10948140,
  author={Wang, Yunbai and Chen, Chen and Chu, Xiaoli},
  journal={IEEE Trans. Wireless Commun.}, 
  title={Performance Analysis for Hybrid Sub-6GHz-mmWave-THz Networks with Downlink and Uplink Decoupled Cell Association}, 
  year={2025},
  volume={},
  number={},
  pages={1-1},
  doi={10.1109/TWC.2025.3555114}}

@article{obaid2026ris,
  title={RIS-Assisted Terahertz Clustered HetNets: Coverage and Rate Analysis},
  author={Obaid, Hadeel and Zhu, Yongxu and Tan, Bo},
  journal={IEEE Open Journal of the Communications Society},
  year={2026},
  publisher={IEEE}
}

\end{document}